\documentclass[11pt]{article}

\usepackage[margin=1in]{geometry}
\usepackage{amsmath,amssymb,amsthm}
\usepackage{enumitem}
\usepackage{algpseudocode}
\usepackage{microtype}
\usepackage[hidelinks]{hyperref}

\newtheorem{theorem}{Theorem}
\newtheorem{lemma}[theorem]{Lemma}

\newcommand{\Z}{\mathbb Z}
\newcommand{\F}{\mathbb F}
\newcommand{\AP}{\operatorname{AP}}
\newcommand{\Merge}{\operatorname{merge}}
\newcommand{\Alg}[1]{\mbox{\textnormal{\textsc{#1}}}}
\newcommand{\Owner}{\mathsf{owner}}
\newcommand{\Dir}{\mathsf{dir}}

\newcounter{algorithm}
\newenvironment{algo}[1]{%
  \refstepcounter{algorithm}\par\medskip
  \noindent\begin{minipage}{\linewidth}\small
  \hrule\smallskip
  \textbf{Algorithm \thealgorithm:} #1\par\smallskip\hrule\smallskip
}{%
  \smallskip\hrule
  \end{minipage}\par\medskip
}

\title{A Linear-Time Algorithm for the Erd\H{o}s--Ginzburg--Ziv
Theorem}
\author{Sunghyeon Jo}
\date{}

\begin{document}
\maketitle

\begin{abstract}
The Erd\H{o}s--Ginzburg--Ziv theorem states that every sequence of
$2n-1$ integers contains a subsequence of length $n$ whose sum is
divisible by $n$.  We give a deterministic algorithm that finds such a
subsequence in $O(n)$ time and $O(n)$ space on a word RAM.  The main
ingredient is an algorithm for the following problem: given $p-1$
nonzero elements of $\F_p$ and a target, find a subset whose sum is the
target.  Reachable sums are encoded by components representing
translated arithmetic progressions.  A table detects the first
collision between a temporary component and an active component.  The
positive collision coefficients are coprime and
allow the two components to be merged into a longer one.  The total
increase in component lengths bounds all table operations and calls to
the Euclidean algorithm.  A reconstruction forest recovers the
required subset, and the product reduction in the original proof
extends the result from prime to arbitrary moduli.
\end{abstract}

\section{Introduction}

Let $A=(a_1,\ldots,a_{2n-1})$ be a sequence of integers.  The theorem
of Erd\H{o}s, Ginzburg, and Ziv~\cite{EGZ1961} guarantees indices
$I\subseteq \{1,\ldots,2n-1\}$ such that
\[
        |I|=n,\qquad \sum_{i\in I} a_i \equiv 0 \pmod n .
\]
The constructive Erd\H{o}s--Ginzburg--Ziv (EGZ) problem is to find
$I$.  Del Lungo, Marini, and Mori
gave a polynomial-time constructive algorithm for finite abelian
groups~\cite{DMM2009}.  Choi, Kang, and Lim~\cite{CKL2022} obtained an
$O(n\log n)$ algorithm for constructive EGZ, and
Leung~\cite{Leung2025} improved the bound to
$O(n\log\log\log n)$.  We prove the following.

\begin{theorem}\label{thm:main}
On a word RAM with $\Theta(\log n)$-bit words, a solution to the
constructive Erd\H{o}s--Ginzburg--Ziv problem can be found
deterministically in $O(n)$ word operations using $O(n)$ words of
working space.
\end{theorem}

The product reduction in the original proof~\cite{EGZ1961} reduces
composite moduli to smaller factors, so it remains to solve the prime
case in linear time.

For a prime $p$, write $\F_p=\Z/p\Z$.  After sorting the input by
residue, either one residue occurs at least $p$ times or the problem
reduces to the following \emph{target subset sum problem over $\F_p$}: given
$d_1,\ldots,d_{p-1}\in\F_p^\times$ and $\tau\in\F_p$, find
$J\subseteq\{1,\ldots,p-1\}$ such that
\[
        \sum_{j\in J}d_j=\tau .
\]
Repeated application of the Cauchy--Davenport
theorem~\cite{Nathanson1996} gives
\[
        \{0,d_1\}+\cdots+\{0,d_{p-1}\}=\F_p .
\]
Thus a solution exists for every target $\tau$.
Our algorithm maintains active components whose sets of nonzero
multiples are pairwise disjoint, together with at most one temporary
component.  It scans the nonzero multiples of the temporary component
in increasing coefficient order.  The first collision with an active
component yields coprime coefficients.  A lemma on bounded integer
representations then allows the two components to be replaced by a
longer one.  Over the whole execution,
the sum of the length increases is $O(p)$; this bounds the work spent
searching for collisions, updating table entries, and computing
inverses.  A forest of reconstruction nodes converts the final
coefficients into a subset of the input indices.  This gives $O(p)$
word operations and $O(p)$ words of working space for target subset sum
over $\F_p$ and, after the product reduction, the same linear bounds for
the original problem.

We use the following computational model.  The word RAM has
$c\log n$-bit words for a sufficiently large
fixed constant $c$.  Input values and intermediate integer sums occupy
$O(1)$ words, and arithmetic on these quantities takes constant time.
Every exact integer used in a component merge has magnitude below
$2n^2$ and fits in one word.  All stated running times count word
operations.

\section{Prime moduli}

We first solve target subset sum over $\F_p$.  Progression elements lie
in $\F_p$, while coefficients and lengths are
ordinary integers.  For $v\in\F_p^\times$ and an integer $\ell\ge0$,
write
\[
        \AP(v,\ell)=\{tv:0\le t\le\ell\}.
\]
For $\ell\le p-1$, the set $\AP(v,\ell)$ has $\ell+1$ distinct
elements.  We call $\ell$ the component length; it counts the positive
coefficients.  The algorithm stores only lengths at most $p-1$; if a
merge produces a larger length, it is truncated to $p-1$.  A component
of length $p-1$ is called \emph{full}; its progression is $\F_p$.

Each component has a fixed identifier.  The algorithm processes
$d_1,d_2,\ldots,d_{p-1}$ from left to right.  The data stored for a
component $C$ represent the translated progression
$\xi(C)+\AP(v(C),\ell(C))$ and satisfy
\[
        v(C)\in\F_p^\times,\qquad 0\le \ell(C)\le p-1,
        \qquad \xi(C)\in\F_p.
\]
The data also include a set $I(C)$ of processed input indices and a
reconstruction-tree root.  For every $0\le t\le\ell(C)$, this root
encodes a subset $R_C(t)\subseteq I(C)$ satisfying
\begin{equation}\label{eq:component-map}
        \sum_{i\in R_C(t)}d_i
        =\xi(C)+t\,v(C)\quad\text{in }\F_p.
\end{equation}

The state contains a family $\mathcal A$ of active components, at most
one temporary component $T$, and a variable $C_{\mathrm{full}}$,
initially $\bot$.  While $C_{\mathrm{full}}=\bot$, write
\[
  \mathcal F=
  \begin{cases}
     \mathcal A, & \text{if $T$ is absent},\\
     \mathcal A\cup\{T\},&\text{if $T$ is present},
  \end{cases}
  \qquad
  X=\sum_{C\in\mathcal F}\xi(C),
  \qquad
  \Lambda=\sum_{C\in\mathcal F}\ell(C).
\]
Once $C_{\mathrm{full}}$ is set to a full component, the main loop
stops.

We identify the nonzero elements of $\F_p$ with their representatives
$1,\ldots,p-1$.  The arrays $\Owner$ and $\Dir$ support
constant-time queries:
\[
  \Owner[1],\ldots,\Owner[p-1]
  \quad\text{and}\quad
  \Dir[1],\ldots,\Dir[p-1].
\]
Here $\Owner[r]=(C,t)$ means that $r=t\,v(C)$ in $\F_p$, and
$\Dir[v]=C$ means that $C$ is the active component with direction $v$.
Empty entries are denoted by $\bot$.  For every active component $C$,
the owner array contains
$(C,t)$ at position $t\,v(C)$ for $1\le t\le\ell(C)$, and
$\Dir[v(C)]=C$.

Replacing the data stored for a component does not change its
identifier, so entries of $\Owner$ and $\Dir$ that refer to it remain
valid.

A temporary component carries two integers
\[
        0\le h(T)\le s(T)\le\ell(T).
\]
Coefficients $1,\ldots,h(T)$ form the \emph{retained prefix}; their
table entries are the \emph{retained owner entries} of $T$.
Coefficients $1,\ldots,s(T)$ form the \emph{scanned prefix}.  Their
residues are known to be disjoint from all active components, and the
owner entries for $h(T)<t\le s(T)$ are empty.  A scan of $T$ therefore
resumes at coefficient $s(T)+1$.

While $C_{\mathrm{full}}=\bot$, the state after each completed
singleton creation, extension, merge, or activation satisfies the
following invariant.
\begin{enumerate}[label=(\roman*),leftmargin=2.8em]
  \item The sets $I(C)$, for $C\in\mathcal F$, are pairwise disjoint
  and partition the processed input indices.  Equation
  \eqref{eq:component-map} holds for every component and coefficient.
  \item The sets $\{t\,v(C):1\le t\le\ell(C)\}$ of the active
  components are pairwise disjoint.
  \item If $T$ exists, the residues
  $\{t\,v(T):1\le t\le s(T)\}$ are disjoint from
  $\{t\,v(C):1\le t\le\ell(C)\}$ for every active component $C$.
  \item The owner array contains exactly the entries of the active
  components and, if $T$ exists, the retained owner entries of $T$.
  Active entries point to their active component.  If $T$ exists, then
  $\Owner[tv(T)]=(T,t)$ for $1\le t\le h(T)$ and
  $\Owner[tv(T)]=\bot$ for $h(T)<t\le s(T)$.  All other owner entries are
  empty.
  \item The direction array contains exactly the active components.
\end{enumerate}

\begin{lemma}\label{lem:invariant}
Whenever this invariant holds, for every
\[
        r\in X+\sum_{C\in\mathcal F}\AP(v(C),\ell(C))
\]
there is a set $J$ of processed input indices such that
$r=\sum_{i\in J}d_i$.  The set $J$ can be recovered from the
reconstruction-tree roots.
\end{lemma}

\begin{proof}
Choose a coefficient $0\le t_C\le\ell(C)$ for each
$C\in\mathcal F$.  By (i), the sets
$R_C(t_C)$ are disjoint, so their union has sum
\[
        \sum_{C\in\mathcal F}
        \bigl(\xi(C)+t_Cv(C)\bigr)
        =X+\sum_{C\in\mathcal F}t_Cv(C).
\]
The reconstruction forest recovers this union.
\end{proof}

\begin{lemma}\label{lem:coprime}
Suppose the invariant holds and the scan of a temporary component $T$
first encounters an occupied owner entry, called a collision, at
coefficient $x$:
\[
        xv(T)=yv(C),\qquad \Owner[xv(T)]=(C,y),
\]
where $C$ is active.  Then $\gcd(x,y)=1$.
\end{lemma}

\begin{proof}
At the collision, $s(T)=x-1$.  If $r=\gcd(x,y)>1$, then
\[
        (x/r)v(T)=(y/r)v(C).
\]
Here $x/r\le s(T)$ and $1\le y/r\le\ell(C)$, contradicting (iii).
\end{proof}

\begin{lemma}\label{lem:frob}
Let $x,y$ be coprime positive integers with $x,y>1$.  Let
$\ell\ge x$ and $m\ge y$, and put
\[
        F=(x-1)(y-1),\qquad S=\ell y+mx.
\]
Then every integer $N$ with
\[
        F\le N\le S-F
\]
has a representation
\[
        N=\alpha y+\beta x,\qquad
        0\le \alpha\le \ell,\quad 0\le \beta\le m .
\]
\end{lemma}

\begin{proof}
Fix a residue class modulo $x$.  Since $\gcd(x,y)=1$, there is a
unique $a_0\in\{0,\ldots,x-1\}$ such that $a_0y$ belongs to this
class.  The nonnegative values of $\alpha$ for which $\alpha y$ lies
in this class are $\alpha=a_0+kx$, where $k\ge0$.  For a fixed $k$,
varying $\beta$ from $0$ to $m$ represents the members of the class in
\[
        [a_0y+kxy,\; a_0y+kxy+mx]\qquad(k\ge0),
\]
restricted to the fixed residue class.  Consecutive intervals overlap
or touch within that class: their starting points differ by $xy$, while
$m\ge y$.

Let $a_1$ be the largest integer at most $\ell$ with
$a_1\equiv a_0\pmod x$.  Thus every integer in the fixed residue class
between $a_0y$ and $a_1y+mx$ is represented with
$0\le\alpha\le\ell$ and $0\le\beta\le m$.

The preceding integer in this residue class is at most
\[
        a_0y-x\le (x-1)y-x=F-1.
\]
Thus every member of the class in $[F,S-F]$ is at least $a_0y$.
For the upper endpoint,
\[
        S-F-(a_1y+mx)=(\ell-a_1)y-(x-1)(y-1)
        \le x-1,
\]
because $0\le\ell-a_1\le x-1$.  The next integer in the class is
$a_1y+mx+x>S-F$.  Hence every member of the class in $[F,S-F]$ lies
between the two represented endpoints, which proves the lemma.
\end{proof}

\begin{lemma}\label{lem:merge}
Let components $C$ and $D$ have disjoint index sets.  Write
$v=v(C)$, $w=v(D)$, $\ell=\ell(C)$, and $m=\ell(D)$.  Suppose there
are positive integers $x\le \ell$ and $y\le m$ such that
\[
        xv=yw\quad\text{in }\F_p ,
\]
and suppose $\gcd(x,y)=1$.  Put
\[
        g=v/y=w/x\quad\text{in }\F_p .
\]
Define the starting coefficient $F$ and raw length $L$ by
\[
 (F,L)=
 \begin{cases}
   (0,m+\ell y),&x=1,\\
   (0,\ell+mx),&x>1\text{ and }y=1,\\
   ((x-1)(y-1),\,mx+\ell y-2(x-1)(y-1)),&x,y>1.
 \end{cases}
\]
Then
\[
        Fg+\AP(g,L)\subseteq \AP(v,\ell)+\AP(w,m).
\]
For every $0\le t\le \min\{L,p-1\}$, there are coefficients
$\alpha$ and $\beta$ satisfying
\[
        0\le\alpha\le\ell,\qquad
        0\le\beta\le m,\qquad
        F+t=\alpha y+\beta x .
\]
For each $t$, fix one such pair and denote it by
$\alpha(t),\beta(t)$.  Consequently, the component $E$ defined by
\[
\begin{split}
 I(E)&=I(C)\mathbin{\dot\cup}I(D),\qquad
 v(E)=g,\qquad \ell(E)=\min\{L,p-1\},\\
 \xi(E)&=\xi(C)+\xi(D)+Fg,\\
 R_E(t)&=R_C(\alpha(t))\mathbin{\dot\cup}R_D(\beta(t))
\end{split}
\]
satisfies \eqref{eq:component-map}.
\end{lemma}

\begin{proof}
Suppose first that $x=1$.  Then $v=yg$ and $w=g$, and $y\le m$ by
assumption.  For each $\alpha=0,\ldots,\ell$, the set
$\alpha yg+\AP(g,m)$ is the interval
$\{\alpha yg,\ldots,(\alpha y+m)g\}$.  Consecutive intervals overlap or
touch because $m\ge y$.  Their union is
$\AP(g,m+\ell y)$.  The remaining case $y=1<x$ is symmetric.

Now assume $x,y>1$.  Since $v=yg$ and $w=xg$,
\[
        \AP(v,\ell)+\AP(w,m)
        =
        \{(\alpha y+\beta x)g:
        0\le\alpha\le\ell,\;0\le\beta\le m\}.
\]
By Lemma~\ref{lem:frob}, every integer
$N\in[F,\ell y+mx-F]$ has a bounded representation
$N=\alpha y+\beta x$.  Taking $N=F+t$ for
$0\le t\le \min\{L,p-1\}$ gives the desired inclusion.

For a fixed $t$, write $\alpha=\alpha(t)$ and $\beta=\beta(t)$.
Equation~\eqref{eq:component-map} for the two children gives
\[
 \sum_{i\in R_C(\alpha)\mathbin{\dot\cup}R_D(\beta)}d_i
 =\xi(C)+\xi(D)+(\alpha y+\beta x)g
 =\xi(E)+tg .
\]
The index sets are disjoint by assumption, so the union is
disjoint.  Hence $E$ satisfies \eqref{eq:component-map}, including after
truncation.
\end{proof}

The reconstruction forest uses leaf, extension, and merge nodes.  A
singleton component $C$ for $d_k$ has direction $d_k$, length $1$,
offset $0$, index set $\{k\}$, and a leaf encoding
$R_C(0)=\varnothing$ and $R_C(1)=\{k\}$.  Extending a component $C$ of
direction $d_k$ produces a component $E$ with the same direction and
offset, length $\ell(C)+1$, and index set
$I(C)\mathbin{\dot\cup}\{k\}$.  The extension node stores $C$'s
reconstruction-tree root and the index $k$, and
\[
 R_E(t)=
 \begin{cases}
   R_C(t),&0\le t\le\ell(C),\\
   R_C(\ell(C))\mathbin{\dot\cup}\{k\},&t=\ell(C)+1.
 \end{cases}
\]
A merge node stores the two child roots, their lengths $\ell$ and $m$,
and $x$, $y$, and $F$; when $x,y>1$, it also stores
$y^{-1}\pmod x$.  The following formulas compute the fixed pair
$\alpha(t),\beta(t)$ in constant time.
If $x=1$, a coefficient
$0\le t\le m+\ell y$ is represented by
\[
        \alpha=\max\left\{0,\left\lceil\frac{t-m}{y}\right\rceil\right\},
        \qquad
        \beta=t-\alpha y .
\]
The inequalities $0\le\alpha\le\ell$ and $0\le\beta\le m$ follow from
$y\le m$ and $t\le m+\ell y$.  The case $y=1$ is symmetric, using
\[
        \beta=\max\left\{0,\left\lceil\frac{t-\ell}{x}\right\rceil\right\},
        \qquad
        \alpha=t-\beta x .
\]

For $x,y>1$, put $N=F+t$.  Let $u\in\{1,\ldots,x-1\}$ be the inverse
of $y$ modulo $x$ stored at the merge node, and compute
\[
        a_0=((N\bmod x)u)\bmod x,\qquad 0\le a_0<x,
        \qquad
        b_0=\frac{N-a_0y}{x}.
\]
This evaluation order keeps every product below $p^2$.  The proof of
Lemma~\ref{lem:frob} gives $b_0\ge0$.  Put
\[
        k=\max\left\{0,\left\lceil\frac{b_0-m}{y}\right\rceil\right\},
        \qquad
        \alpha=a_0+kx,\qquad
        \beta=b_0-ky .
\]
Then $N=\alpha y+\beta x$, and the proof of Lemma~\ref{lem:frob}
gives $0\le\alpha\le\ell$ and
$0\le\beta\le m$.  Once the inverse is stored at the merge node, this
conversion takes constant time.

For a temporary component $T$ colliding with an active component $C$ at
coefficients $x$ and $y$, write $\Merge(T,C;x,y)$ for the component
constructed by Lemma~\ref{lem:merge}, with $T$ first and $C$ second.
Lemma~\ref{lem:coprime} gives $\gcd(x,y)=1$.

Fix nonzero elements $d_1,\ldots,d_{p-1}\in\F_p^\times$ and a target
$\tau\in\F_p$.  The following algorithm, \Alg{PrimeTarget}, returns a set
$J\subseteq\{1,\ldots,p-1\}$ such that
$\sum_{j\in J}d_j=\tau$.  Whenever it creates a temporary component,
it calls the subroutine \Alg{Normalize} in Algorithm~\ref{alg:normalize}.

\begin{algo}{\Alg{PrimeTarget}$(p;d_1,\ldots,d_{p-1};\tau)$}
\label{alg:prime-target}
\begin{algorithmic}[1]
\State $\mathcal A\gets\varnothing$ and $\Lambda\gets0$; initialize
$\Owner$ and $\Dir$ to $\bot$; set $T,C_{\mathrm{full}}\gets\bot$
\State Build the table of inverses in $\F_p$
\For{$k=1,\ldots,p-1$}
  \If{$\Dir[d_k]=\bot$}
    \State Let $T$ be a singleton component for $d_k$ and set
    $h(T),s(T)\gets0,0$; set $\Lambda\gets\Lambda+1$
    \State $\Call{Normalize}{T}$
  \Else
    \State $C\gets\Dir[d_k]$, $\ell'\gets\ell(C)+1$; let $E$ be the
    extension of $C$ by $d_k$; set $\Lambda\gets\Lambda+1$
    \If{$\ell'=p-1$}
      \State Remove $C$ from $\mathcal A$, clear its owner entries,
      and set $\Dir[d_k]\gets\bot$
      \State Replace $C$'s stored data by $E$'s data; set
      $C_{\mathrm{full}}\gets C$ and \textbf{break}
    \ElsIf{$\Owner[\ell'd_k]=\bot$}
      \State Replace $C$'s stored data by $E$'s data and set
      $\Owner[\ell'd_k]\gets(C,\ell')$
    \Else
      \State Remove $C$ from $\mathcal A$, set $\Dir[d_k]\gets\bot$,
      and clear its owner entries
      \State Replace $C$'s stored data by $E$'s data; set $T\gets C$,
      $h(T)\gets0$, and $s(T)\gets\ell'-1$
      \State $\Call{Normalize}{T}$
    \EndIf
  \EndIf
  \If{$C_{\mathrm{full}}\ne\bot$ or $\Lambda=p-1$}
    \State \textbf{break}
  \EndIf
\EndFor
\State Set $\mathcal G\gets\mathcal A$; if $C_{\mathrm{full}}\ne\bot$,
set $\mathcal G\gets\mathcal A\cup\{C_{\mathrm{full}}\}$
\State Set $X\gets\sum_{C\in\mathcal G}\xi(C)$ and
$t_C\gets0$ for every $C\in\mathcal G$
\If{$C_{\mathrm{full}}\ne\bot$}
  \State Choose the unique
  $t_{C_{\mathrm{full}}}\in\{0,\ldots,p-1\}$ such that
  $t_{C_{\mathrm{full}}}v(C_{\mathrm{full}})=\tau-X$
\ElsIf{$\tau-X\ne0$}
  \State Read $\Owner[\tau-X]=(C,t)$ and set $t_C\gets t$
\EndIf
\State \Return $\displaystyle\bigcup_{C\in\mathcal G}R_C(t_C)$
\end{algorithmic}
\end{algo}

The subroutine \Alg{Normalize}$(T)$ assumes the invariant above.  It
updates the shared state and returns after either activating the current
temporary component or storing it in $C_{\mathrm{full}}$; in either case
it sets $T\gets\bot$.

\begin{algo}{\Alg{Normalize}$(T)$}
\label{alg:normalize}
\begin{algorithmic}[1]
\While{$\ell(T)<p-1$}
  \While{$s(T)<\ell(T)$ and
  $\Owner[(s(T)+1)v(T)]=\bot$}
    \State $s(T)\gets s(T)+1$
  \EndWhile
  \If{$s(T)=\ell(T)$}
    \For{$t=h(T)+1,\ldots,\ell(T)$}
      \State $\Owner[tv(T)]\gets(T,t)$
    \EndFor
    \State $\mathcal A\gets\mathcal A\cup\{T\}$ and
    $\Dir[v(T)]\gets T$
    \State $T\gets\bot$
    \State \Return
  \EndIf
  \State Set $x\gets s(T)+1$; let $\Owner[xv(T)]=(C,y)$ and
  $m\gets\ell(C)$
  \State $E\gets\Merge(T,C;x,y)$
  \State Remove $C$ from $\mathcal A$ and set $\Dir[v(C)]\gets\bot$
  \State $\Lambda\gets\Lambda+\ell(E)-\ell(T)-m$
  \If{$x=1$}
    \State Clear the retained owner entries of $T$
    \State Replace $C$'s stored data by $E$'s data, retaining $C$'s
    owner entries
    \State Set $T\gets C$ and $h(T),s(T)\gets m,m$
  \ElsIf{$y=1$}
    \State Let $h_0\gets h(T)$ and clear the owner entries of $C$
    \State Replace $T$'s stored data by $E$'s data, retaining $T$'s
    owner entries
    \State Set $h(T),s(T)\gets h_0,x$
  \Else
    \State Clear the retained owner entries of $T$ and all owner
    entries of $C$
    \State Replace $T$'s stored data by $E$'s data and set
    $h(T),s(T)\gets0,0$
  \EndIf
\EndWhile
\State $C_{\mathrm{full}}\gets T$
\State $T\gets\bot$
\State \Return
\end{algorithmic}
\end{algo}

\begin{lemma}\label{lem:target-correct}
\Alg{PrimeTarget} returns a subset $J\subseteq\{1,\ldots,p-1\}$
such that $\sum_{j\in J}d_j=\tau$.
\end{lemma}

\begin{proof}
The empty state satisfies the invariant, and a singleton component
satisfies \eqref{eq:component-map}.  For an extension, $R_E$ is valid
because $d_k=v(C)$, and it adds the index $k$.  If
$\Owner[\ell'd_k]$ is empty, writing it preserves disjointness.
Otherwise, clearing $C$'s old owner entries makes the extension
temporary; its old residues remain disjoint from the nonzero multiples
of the active components.

Lemma~\ref{lem:merge} gives equation~\eqref{eq:component-map} for the
merged component and combines two disjoint index sets.  If $x=1$, the
new direction is $v(C)$.  For $1\le t\le m$, reconstruction uses only
the $C$-child, so its owner entries remain; their residues were
disjoint from the nonzero multiples of every other active component.
If $y=1<x$, the new direction is $v(T)$.  Clearing the entries of $C$
leaves the scanned prefix valid and makes $\Owner[xv(T)]$ empty, so
coefficient $x$ can be added.  If $x,y>1$, resetting both scan
counters to zero starts the next scan at coefficient $1$.  Updating
the indicated owner entries and removing $C$ from the direction array
gives (iii)--(v).  If a scan finishes, every unwritten owner entry of
the temporary component is empty and can be written.  Thus every update
made while $C_{\mathrm{full}}=\bot$ preserves the invariant.

If $C_{\mathrm{full}}$ is set, it has length $p-1$.  Since
$\AP(v(C_{\mathrm{full}}),p-1)=\F_p$, coefficient $0$ may be chosen
in every other component, and a unique coefficient of
$C_{\mathrm{full}}$ makes the represented sum equal to $\tau$.

Suppose that no full component is created.  Each iteration ends with
all components active: an extension whose new owner entry is empty
remains active, and \Alg{Normalize} returns only after activating its
temporary component.  Processing a new input index increases $\Lambda$
by one.  Every non-full merge has new length at least the sum of the two
lengths it replaces.  The formulas for $x=1$ and $y=1$ show this
directly; if $x,y>1$, then
\[
 L-\ell-m=(m-y)(x-1)+(\ell-x)(y-1)+x+y-2\ge0.
\]
Consequently, after processing all $p-1$ inputs, or earlier when the stopping test in
\Alg{PrimeTarget} succeeds, the final total length is at
least $p-1$.  By (ii), the
nonzero multiples covered by the active components are pairwise
distinct elements of $\F_p^\times$, so the total length is at most
$p-1$.  They therefore cover all of $\F_p^\times$.  If $\tau-X=0$,
choose coefficient $0$ in every component; otherwise use the unique
owner entry $\Owner[\tau-X]=(C,t)$.

Without a full component, Lemma~\ref{lem:invariant} converts the chosen
coefficients into a set of input indices whose corresponding values
sum to $\tau$.  With a full component, the same union argument applies
to $\mathcal G$, whose index sets remain disjoint.
\end{proof}

We now analyze the running time and space.  All offsets are stored as
least nonnegative residues modulo $p$.  Precompute the field inverses
$\operatorname{inv}_p(a)=a^{-1}\pmod p$
for $1\le a<p$.  In increasing order they satisfy
\[
 \operatorname{inv}_p(1)=1,\qquad
 \operatorname{inv}_p(a)\equiv
 \left(p-\left\lfloor\frac pa\right\rfloor\right)
 \operatorname{inv}_p(p\bmod a)\pmod p
 \quad(2\le a<p).
\]
Because $p$ is prime, $1\le p\bmod a<a$ in the second formula.  The
table therefore takes $O(p)$ word operations and $O(p)$ words of
working space,
and supplies all field divisions used for the directions of merged
components and for target recovery.

For $x,y>1$, the extended Euclidean algorithm computes
$y^{-1}\pmod x$, which is stored in the merge node.  The formulas for
$x=1$ and $y=1$ use integer division.

When $x,y>1$, we have $x\le\ell\le p-1$ and
$y\le m\le p-1$; hence the exact merge quantities satisfy
\[
  F\le(p-2)^2,\qquad
  S=\ell y+mx<2p^2,\qquad L=S-2F<2p^2.
\]
If $x=1$ or $y=1$, the raw length $L$ is less than $p^2$.  These
quantities are evaluated exactly before the component length is
truncated to $p-1$.  During reconstruction,
$N=F+t<p^2$.  The new offset is evaluated modulo $p$ as
\[
 \xi(E)=\bigl(\xi(T)+\xi(C)+(F\bmod p)g\bigr)\bmod p
\]
for a merge of a temporary component $T$ and an active component $C$.

A scan reads one owner entry for each positive coefficient it examines.
If the scan reaches the end of a temporary component, the empty owner
entries are written once when that component becomes active.

\begin{lemma}\label{lem:amortized}
The algorithm performs $O(p)$ owner-array reads and writes.
\end{lemma}

\begin{proof}
Recall that
\[
        \Lambda=\sum_{C\in\mathcal F}\ell(C).
\]
When $T$ is present, its retained prefix is already included in
$\ell(T)$.  As maintained by \Alg{PrimeTarget} and
\Alg{Normalize}, processing a new input index increases $\Lambda$ by
one, and a merge changes it by the difference between the new and old
component lengths.

Assign two units of credit to each positive coefficient of a temporary
component whose owner entry is empty: one for reading its owner entry
and one for writing that entry if the scan finishes.  A new singleton
receives this credit when its input is processed.  An extension whose new
owner entry is empty uses constant work for that entry.  After a merge,
the credits of the old temporary coefficients and any increase in
$\Lambda$ pay for the
remaining work and the credits assigned to the new temporary
coefficients.  The scan skips retained owner entries.

For a merge of components of lengths $\ell$ and $m$, write
$\Delta=L-\ell-m$.  If extending an active component of old length
$\ell_0$ exposes a collision, \Alg{PrimeTarget} clears its $\ell_0$
owner entries before the merge.
Here the temporary length is $\ell=\ell_0+1$ and the collision
coefficient is $x=\ell$.  For a non-full merge with $y=1$,
$\Delta=m(x-1)\ge\ell_0$; if $y>1$, then
$\Delta\ge\ell+m-2\ge\ell_0$.  Hence this clearing is charged to the
immediately following merge; the same charge covers the one possible
rewrite of those entries when the temporary component becomes active.
An extension that produces a full component clears at most $p-2$ owner
entries, and no further scan occurs.

Consider a non-full merge of a temporary component of length $\ell$ and
an active component of length $m$.

If $x,y>1$, the new length is
\[
        L=mx+\ell y-2(x-1)(y-1).
\]
The increase in length is
\[
        \Delta=L-\ell-m
        =(m-y)(x-1)+(\ell-x)(y-1)+x+y-2 .
\]
Since $\ell\ge x$ and $m\ge y$, this is positive.  Moreover
\[
        \Delta-(\ell+m-2)=(m-y)(x-2)+(\ell-x)(y-2)\ge0 .
\]
Thus the scan that finds the collision, deletion of the active
component's owner entries, later writing of the new owner entries,
and the credit assigned to the new temporary component cost
$O(\ell+m+L)=O(\Delta)$.  The reads performed before the collision are
already paid for by the corresponding credits; those credits are then
discarded.

If $x=1$ and $y>1$, the active component has direction $g$ and length
$m$, while the temporary component has direction $yg$.  The merged
component has direction $g$ and length $m+\ell y$, and retains the active
component's $m$ owner entries.  Clearing the retained owner entries of
the temporary component and processing the suffix $m+1,\ldots,m+\ell y$ costs
$O(\ell+\ell y)=O(\Delta)$, since $\Delta=\ell(y-1)$.

If $y=1$ and $x>1$, the temporary component has direction $g$, while the
active component has direction $xg$.  The merged component has direction
$g$ and length $\ell+mx$.  The scanned coefficients of the temporary
component remain valid, and their work is already charged.  Deleting the
active component's $m$ owner entries and processing the suffix of length
$mx$ costs $O(mx+m)=O(\Delta)$, since $\Delta=m(x-1)$.

If $x=y=1$, the total length is unchanged.  The active component's
owner entries are retained.  For $1\le j\le\ell$, suffix coefficient
$m+j$ uses temporary coefficient $j$, whose credit pays for the new
owner entry.  Thus repeated merges with $x=y=1$ charge distinct
coefficients.

Processing new input indices increases $\Lambda$ by one.  Merges with
$x=y=1$ preserve it, and every other non-full merge increases it by the
positive $\Delta$ computed above.  Until a full component is produced,
the active components cover pairwise disjoint sets of nonzero multiples,
of total size at most $p-1$; a temporary component, if present, has
length at most $p-2$.

Hence
$\Lambda<2p$ during calls to \Alg{Normalize}, so the sum of all positive
increases is $O(p)$.  If an
extension creates a full component or a merge has $L\ge p-1$, this is
the final update; its scan and deletions cost $O(p)$, after which the
owner and direction arrays may be discarded.  Thus all owner-array
work is linear.
\end{proof}

\begin{theorem}\label{thm:prime-target}
In the word RAM model specified in the Introduction,
\Alg{PrimeTarget} solves target subset sum over $\F_p$ in $O(p)$ word
operations and $O(p)$ words of working space.
\end{theorem}

\begin{proof}
Correctness follows from Lemma~\ref{lem:target-correct}.  The field-inverse table
takes $O(p)$ word operations and $O(p)$ words of working space, and
Lemma~\ref{lem:amortized} bounds the owner-array work by $O(p)$.

Processing each $d_k$ creates one leaf or extension node.  A merge
decreases the number of current reconstruction-tree roots by one, so
after processing $k$ inputs there are
at most $k-1$ merge nodes.  Thus the algorithm performs $O(p)$ merges
and creates $O(p)$ reconstruction nodes.

Collisions with $x=1$ or $y=1$ require $O(1)$ arithmetic each and
therefore $O(p)$ time in total.
For $x,y>1$, the Euclidean algorithm computes $\gcd(x,y)$ and
$y^{-1}\pmod x$ in $O(\log\min\{x,y\})$ word operations.  For every
such merge that does not produce a full component, the length increase
$\Delta=L-\ell-m$ satisfies
\[
        \Delta\ge x+y-2,
\]
so its Euclidean cost is $O(\Delta)$.  The sum of the positive
increases in $\Lambda$ is $O(p)$, and a merge that produces a full
component contributes at most one additional $O(\log p)$ computation.
Thus all Euclidean
computations take $O(p)$ time.

The reconstruction nodes form a forest: an extension or merge places
each old root below a new root and removes the old root from the current
family.  Reconstruction visits each node below the final roots,
including those assigned coefficient $0$, at most once.  Leaf and extension
nodes take constant time, and a merge node computes its two child
coefficients in constant time from the stored data.  The reconstruction
therefore takes $O(p)$ time and $O(p)$ words of working space.

All remaining bookkeeping takes $O(p)$ word operations.  The arrays,
components, and reconstruction nodes together occupy $O(p)$ words.
\end{proof}

We can now solve the constructive EGZ problem for prime moduli.  We
denote the resulting algorithm by
\Alg{PrimeEGZ}$(p;a_1,\ldots,a_{2p-1})$.

\begin{theorem}\label{thm:prime-egz}
Given a prime $p$ and integers $a_1,\ldots,a_{2p-1}$,
\Alg{PrimeEGZ} finds $p$ input indices whose values sum to $0$ modulo
$p$ in $O(p)$ word operations and $O(p)$ words of working space.
\end{theorem}

\begin{proof}
\Alg{PrimeEGZ} first reduces the input to $\F_p$ and sorts the indices
by residue using counting sort.  Let $\pi$ be the resulting permutation and put
$b_j=a_{\pi(j)}$ in $\F_p$.  If
$b_i=b_{i+p-1}$ for some $1\le i\le p$, then the $p$ indices
\[
        \{\pi(j):i\le j\le i+p-1\}
\]
have a common residue and therefore sum to $0$ modulo $p$.

Suppose that no such block exists.  Every residue occurs at most
$p-1$ times.  For
$1\le i\le p-1$, set $d_i=b_{i+p}-b_i$.  Since
$b_i\ne b_{i+p}$ in $\F_p$, every $d_i$ is nonzero.  Put
\[
        \tau=-\sum_{j=1}^p b_j\quad\text{in }\F_p.
\]
\pagebreak[3]
By Theorem~\ref{thm:prime-target},
\Alg{PrimeTarget}$(p;d_1,\ldots,d_{p-1};\tau)$ returns
$J\subseteq\{1,\ldots,p-1\}$ such that
$\sum_{i\in J}d_i=\tau$.  Starting from the index set
$\{\pi(1),\ldots,\pi(p)\}$, replace $\pi(i)$ by $\pi(i+p)$ for each
$i\in J$.  The resulting set has exactly $p$ original indices, with sum
\[
        \sum_{j=1}^p b_j+\sum_{i\in J}(b_{i+p}-b_i)
        \equiv \sum_{j=1}^p b_j+\tau
        \equiv 0\pmod p.
\]
Counting sort uses $O(p)$ word operations and $O(p)$ words of working
space, and Theorem~\ref{thm:prime-target} gives the same bounds for the
remaining work.
\end{proof}

\section{Composite moduli}

For arbitrary $n$, we use the product reduction from the original
proof~\cite{EGZ1961}.

\begin{proof}[Proof of Theorem~\ref{thm:main}]
We define \Alg{EGZ}$(n;a_1,\ldots,a_{2n-1})$ recursively.  If $n=1$,
it returns the only input index.  If $n$ is prime, it returns
\Alg{PrimeEGZ}$(n;a_1,\ldots,a_{2n-1})$.  Suppose that $n$ is
composite.  Choose a prime divisor $p$ and write $n=pq$.

Initialize a pool with all $2pq-1$ input indices.  For each
$h=1,\ldots,2q-1$, remove a set $Q_h$ of any $2p-1$ indices from the
pool and apply \Alg{PrimeEGZ} to the values indexed by $Q_h$.  Let
$B_h\subseteq Q_h$ be the resulting set of $p$ indices; their values
sum to a multiple of $p$.  Keep $B_h$
outside the pool, return $Q_h\setminus B_h$ to the pool, and define
\[
        c_h=\frac{1}{p}\sum_{i\in B_h}a_i .
\]
Let $H$ be the set of $q$ indices returned by
\Alg{EGZ}$(q;c_1,\ldots,c_{2q-1})$.
The sets $B_h$ are disjoint, and
$\sum_{h\in H}c_h\equiv0\pmod q$.  Hence
$\bigcup_{h\in H}B_h$ contains $pq=n$ original indices, and their sum is
\[
        p\sum_{h\in H}c_h,
\]
which is divisible by $pq$.

Each $c_h$ is an integer, and
$|c_h|\le\max_{i\in B_h}|a_i|$.  Thus recursive values satisfy the
same $O(1)$-word input assumption.

The pool always contains enough indices for the next call to
\Alg{PrimeEGZ}.  Immediately before iteration $h$ it has
\[
        2pq-1-(h-1)p
\]
indices, because every earlier iteration permanently removed exactly
$p$.  For $h\le2q-1$ this quantity is at least $2p-1$.

Let $T(n)$ be the running time.  Trial division determines whether $n$
is prime and, when it is composite, finds a prime divisor in
$O(\sqrt n)$ word operations.  This is included in the $O(n)$ work at
that recursion level.  The loop that constructs the blocks uses a
single pool of unused indices: each call takes $2p-1$ indices and
returns the unselected $p-1$.  Each call
performs $O(p)$ pool operations, and the $2q-1$ calls perform $O(pq)$
in total at this recursion level.  If $n=pq$ with $p$ prime, the
composite step makes $2q-1$ calls
to \Alg{PrimeEGZ} on modulus $p$, one recursive call on $q$, and
performs $O(pq)$ additional bookkeeping.  By Theorem~\ref{thm:prime-egz},
\[
        T(pq)=T(q)+O(pq).
\]
Since $q\le n/2$ for composite $n$, iterating the recurrence gives
\[
        T(n)=O\!\left(n+\frac n2+\frac n4+\cdots\right)=O(n).
\]
At a level of size $n$, the stored blocks contain
$p(2q-1)<2n$ indices, and the pool and one call to \Alg{PrimeEGZ} use
$O(n)$ more words of working space.  The working space for
\Alg{PrimeEGZ} is reused between calls.  If
$S(n)$ denotes peak working space, then
\[
        S(n)\le S(q)+O(n),
\]
which also sums geometrically because $q\le n/2$.  Thus the working
space is $O(n)$ words.
\end{proof}


\begin{thebibliography}{9}

\bibitem{CKL2022}
S. Choi, H. Kang, and D. Lim.
\newblock Simple deterministic $O(n\log n)$ algorithm finding a solution
of Erd\H{o}s--Ginzburg--Ziv theorem.
\newblock \emph{arXiv:2208.07728}, 2022.

\bibitem{DMM2009}
A. del Lungo, C. Marini, and E. Mori.
\newblock A polynomial-time algorithm for finding zero-sums.
\newblock \emph{Discrete Mathematics}, 309(9):2658--2662, 2009.
\newblock \href{https://doi.org/10.1016/j.disc.2008.06.018}
{doi:\,10.1016/j.disc.2008.06.018}.

\bibitem{EGZ1961}
P. Erd\H{o}s, A. Ginzburg, and A. Ziv.
\newblock Theorem in the additive number theory.
\newblock \emph{Bulletin of the Research Council of Israel, Section F:
Mathematics and Physics}, 10F:41--43, 1961.

\bibitem{Leung2025}
Y. H. A. Leung.
\newblock Finding a solution to the Erd\H{o}s--Ginzburg--Ziv theorem in
$O(n\log\log\log n)$ time.
\newblock \emph{arXiv:2507.08139}, 2025.

\bibitem{Nathanson1996}
M. B. Nathanson.
\newblock \emph{Additive Number Theory: Inverse Problems and the
Geometry of Sumsets}.
\newblock Graduate Texts in Mathematics 165, Springer, 1996.

\end{thebibliography}
\end{document}